\documentclass[10pt,showpacs,amsmath,amssymb]{article}
\usepackage{amsfonts}
\usepackage{amsmath}
\usepackage{amsthm}
\usepackage{color}

\newtheorem{theorem}{Theorem}[section]
\newtheorem{lemma}[theorem]{Lemma}

\newcommand{\field}[1]{\mathbb{#1}}

\numberwithin{equation}{section}

\begin{document}

\title{Constant curvature holomorphic solutions of the 
supersymmetric grassmannian sigma model: the case of $G(2,4)$}

\author{\vspace{1cm}\\
        {\bf V. Hussin}$^{1,2}$
        \thanks{E-mail address:
        hussin@dms.umontreal.ca} \,,
       {\bf M. Lafrance}$^{2}$
        \thanks{E-mail address:
        marie.lafrance@umontreal.ca}\, and 
         {\bf \.{I}. Yurdu\c{s}en}$^{3}$
        \thanks{E-mail address:
        yurdusen@hacettepe.edu.tr}
 \\
\\$^1$Centre de Recherches Math\'{e}matiques, 
Universit\'{e} de Montr\'{e}al,\\ CP 6128, Succ. Centre-Ville, Montr\'{e}al, 
Quebec H3C 3J7, Canada
\\
\\$^2$D\'{e}partement de Math\'{e}matiques et de Statistique, 
Universit\'{e} de Montr\'{e}al,\\ CP 6128, Succ. Centre-Ville, Montr\'{e}al, 
Quebec H3C 3J7, Canada
\\
   \\$^3$Department of Mathematics, Hacettepe University,
                    \\ 06800 Beytepe, Ankara, Turkey
                       }

\date{\today}

\maketitle

\begin{abstract}

We explore the constant curvature holomorphic solutions of the supersymmetric  grassmannian sigma model $G(M,N)$ using in particular the gauge invariance of the model.  Supersymmetric invariant solutions are constructed via generalizing a known result for $\field{C}P^{N-1}$. We show that some other such solutions also exist. Indeed, considering the simplest case of $G(2,N)$ model, we give necessary and sufficient conditions for getting the constant curvature holomorphic solutions. Since, all the constant curvature holomorphic solutions of the bosonic $G(2,4)$ sigma model are known, we treat this example in detail.

\end{abstract}



\section{Introduction \label{intro}}
The Weierstrass representation of surfaces in multidimensional spaces
\cite{KonopelchenkoLandolfi, KonopelchenkoTaimanov1996, CK1996, Heleinbook}, 
such as Lie algebras and groups, has generated interest in studying surfaces associated with 
the solutions of the grassmannian bosonic $G(M,N)$ sigma model ($\sigma$-model) \cite{GY1, HYZ, PostG}. 
Motivated by the work dealing with $G(2, 4)$ \cite{4solutionsG(24)} and $G(2, 5)$
 \cite{solutionsG(25)}, a general 
approach for constructing holomorphic maps of 2-sphere $S^2$ of constant curvature into 
$G(M, N)$ have been realized in two papers 
\cite{ DHZ2013I, DHZ2}. 


Then in \cite{DHYZ, DHZ5}, most of the above ideas have been generalized to the 
supersymmetric (susy) $\field{C}P^{N-1}$ $\sigma$-model which is equivalent to the susy  $G(1, N)$.
In particular, all the susy invariant solutions with constant curvature 
of this model have been thoroughly discussed.

The natural question is
to extend those results to susy $G(M, N)$ $\sigma$-models for $M>1$. In order to 
achieve some canonical results we use the full power of the 
gauge invariance of these susy models. 
Indeed recently, present authors explored 
such type of invariance  \cite{HLYZ}. Although it
is well-known \cite{witten, susyG(MN), fujii1, fujii2, Zbook}, 
to our knowledge, up to now no explicit form of it had been
used in an effective way to analyse the solutions of the model. 
In particular, the gauge invariance of 
the susy model is much richer than its bosonic counterpart.
We will thus use it in the present paper in order 
to construct constant curvature holomorphic solutions of the susy 
$G(2,4)$ $\sigma$-model.

The structure of the article is as follows: In 
Section \ref{constancurvature}, we discuss the necessary and sufficient conditions to get the constant curvature holomorphic solutions of the general susy $G(M,N)$ $\sigma$-model. 
In Section \ref{examples} we give a detailed analysis of the susy $G(2,4)$ $\sigma$-model. Taking into account the susy gauge invariance we present all the holomorphic solutions of this 
model with constant curvature in a canonical form. In Section \ref{casecp5}, a well-known embedding of $G(2,4)$ into $\field{C}P^{5}$ is given in order to help to understand some arbitrariness in the susy solutions of $G(2,4)$.  Finally, we end
the article by giving some comments in Section \ref{conc}.


\section{Constant curvature holomorphic solutions of the susy $G(M,N)$ $\sigma$-model 
\label{constancurvature}}


For the susy $G(M,N)$ $\sigma$-model \cite{Zbook},  a general bosonic superfield $\Phi:\tilde{S}^2\mapsto G(M,N)$ has the following expansion
\begin{eqnarray}
\Phi(x_{\pm},\theta_{\pm}) =  \Phi_0(x_{\pm}) + 
i \theta_{+} \Phi_{1}(x_{\pm}) +
i \theta_{-} \Phi_{2}(x_{\pm}) - 
\theta_{+}\theta_{-} \Phi_3(x_{\pm}),
\label{intsuperbosefield}
\end{eqnarray}
where $\Phi_0$ and $\Phi_3$ are $N \times M$ bosonic complex matrices and 
$\Phi_1$ and $\Phi_2$ are $N \times M$ fermionic complex matrices. Here, $\tilde{S}^2$ 
denotes the superspace $(x_{\pm}, \theta_{\pm})$ whose bosonic part is compactified 
to $2$-sphere $S^2$. This bosonic superfield must satisfy
\begin{eqnarray}
\Phi^{\dagger} \Phi = \mathbb{I}_M.  \label{norm}
\end{eqnarray}

As in the bosonic case, holomorphic solutions of the susy $G(M,N)$ $\sigma$-model are trivial solutions of the model \cite{HLYZ, Zbook}. It has been shown that they take the form
\begin{eqnarray}
\Phi = W L,
\label{macfarlanesusy}
\end{eqnarray}
where $W$ is an $N\times M$ matrix depending only on the coordinates $(x_+,\theta_+)$ while $L$  is a non singular $M\times M$ matrix that depends on the coordinates 
$(x_\pm,\theta_\pm)$. It means that  the holomorphic superfield $W$ takes the explicit form 
\begin{eqnarray}
W(x_+, \theta_+) = Z(x_+) + i \theta_+ \eta(x_+) A(x_+)\,
\label{MFparasusy}
\end{eqnarray}
and the determination of the holomorphic solutions of the susy $G(M,N)$ $\sigma$-model is equivalent \cite{HLYZ, Zbook} to the study of these holomorphic superfields.

In the case of the susy $\field{C}P^{N-1}$ $\sigma$-model, the solutions 
of the susy Euler-Lagrange equations 
have been shown to be associated with surfaces \cite{Hussin2006}. 
The susy Gaussian curvature of the surface 
corresponding to the susy holomorphic solution $W$ 
was given by the formula 
\begin{eqnarray}
\tilde{\kappa} = -\frac{1}{\tilde{g}} \partial_+ \partial_- \ln \tilde{g}, 
\label{curvsusy}
\end{eqnarray}
where the supersymmetric expression of the metric was
\begin{eqnarray}
\tilde{g} = \partial_+ \partial_- \ln\left(\det (W^\dagger W)\right).
\label{metricsusy}
\end{eqnarray}
Clearly the metric and curvature may be functions of $(x_{\pm}, \theta_{\pm})$ even if  $W$ depends only on the coordinates $(x_+,\theta_+)$. 

For the case of susy $G(M,N)$, we assume the same relation between the superfield $W$, metric $\tilde{g}$ and curvature $\tilde{\kappa}$. It means that asking for a constant curvature solution is 
equivalent to assuming that $\tilde{\kappa} = \kappa$ where 
$\kappa$ is a purely bosonic constant (a strictly positive real number) and 
must be the  curvature associated with the purely 
bosonic $G(M,N)$ solution $Z$ involved in $W= (\ref{MFparasusy})$. 


Let us write explicitly the condition (\ref{curvsusy}) using the expression of $W$ in (\ref{MFparasusy}) and taking into account that  $\tilde{\kappa} = \kappa$. In order to simplify the calculations, we take
\begin{eqnarray}
T_1 = \theta_+ \eta, \quad  T_2 = \theta_- \eta^{\dagger}. 
\label{T}
\end{eqnarray}
Notice that since $T_1$ and $T_2$ are both product of 
two fermionic functions, we have $T_1^2=0$ and $T_2^2=0$. Moreover, 
they are bosonic quantities and hence commute with all the other quantities.

We thus easily get
\begin{eqnarray}
\det \left(W^\dagger W\right) &=& (\det M_0) \det \left(\mathbb{I}_M 
+ i T_1 M_0^{-1} M_1 + i T_2 M_0^{-1} M_2 - T_1 T_2 M_0^{-1} M_3\right) \nonumber \\
&=& (\det M_0) \left(1 + i T_1 X_1 + i T_2 X_2 - T_1 T_2 X_3\right),
\label{detWdagW}
\end{eqnarray}
with 
\begin{eqnarray}
M_0 = Z^{\dagger} Z\,, \quad M_1 = Z^{\dagger} A\,, \quad M_2 = A^{\dagger} Z\,, \quad M_3 = A^{\dagger} A\,. 
\label{reviseddefineMV}
\end{eqnarray}
The expressions of $X_1, X_2$ and $X_3$ remain to be explicitly computed.

The metric $\tilde{g}=$ (\ref{metricsusy}) takes the form
\begin{eqnarray}
\tilde{g} = g
+ \partial_+ \partial_- \ln \left(1 + i T_1 X_1 + i T_2 X_2 - T_1 T_2 X_3 \right),
\label{metricexpanded}
\end{eqnarray}
with 
\begin{eqnarray}
g = \partial_+ \partial_- \ln\left(\det (Z^\dagger Z)\right).
\label{metricbos}
\end{eqnarray}
Using the Taylor expansion of the logarithmic function 
\begin{eqnarray}
\ln \left(1 + x \right) = x - \frac{x^2}{2} + \mathcal{O}\hspace{0.2mm}(x^3),
\label{Taylorlog}
\end{eqnarray}
we get
\begin{eqnarray}
\tilde{g}= g+ 
\partial_+ \partial_- \left[i T_1 X_1 + i T_2 X_2 - T_1 T_2 \left(X_3-X_1X_2\right)\right].
\label{metricsimplified}
\end{eqnarray}

By a similar procedure we can express the quantity 
$\partial_+ \partial_- \ln \tilde{g}$ as;
\begin{eqnarray}	
\partial_+ \partial_- \ln \tilde{g} &=& 
\partial_+ \partial_- \ln g
+ i T_1\, \partial_+ \partial_-\, Y_1 + i T_2\, \partial_+ \partial_-\, Y_2 \nonumber \\ 
&-& T_1\, T_2\, \partial_+ \partial_-\,\big(Y_3 - Y_1 Y_2\big)\,,
\label{mixedderoflog}
\end{eqnarray}	 
with 
\begin{eqnarray}
&&Y_1 \equiv \frac{\kappa}{2} (1+|x|^2)^2 \partial_+ \partial_- X_1\,, \nonumber \\
&&Y_2 \equiv \frac{\kappa}{2} (1+|x|^2)^2 \partial_+ \partial_- X_2\,, \nonumber \\
&&Y_3 \equiv \frac{\kappa}{2} (1+|x|^2)^2 \partial_+ \partial_- \big(X_3 - X_1\,X_2\big)\,. 
\label{definitionY123}
\end{eqnarray}	

Upon inserting these relations into (\ref{curvsusy}) we get the following constraints 
\begin{equation}
\partial_+ \partial_- \ln g+ \kappa g=0
\label{condg} 
\end{equation}
and
\begin{eqnarray}
&&\partial_+ \partial_- \left(Y_1 +\kappa X_1\right) = 0\,, 
\qquad \partial_+ \partial_- \left(Y_2 + \kappa X_2\right) = 0\,, \label{necsufconditions12} \\
&&\partial_+ \partial_- \left((Y_3 - Y_1\,Y_2) + \kappa (X_3 - X_1\,X_2)\right) = 0\,. 
\label{necsufconditions3} 
\end{eqnarray}

Notice that the two expressions in (\ref{necsufconditions12}) are complex 
conjugate to each other and hence we have only one independent condition, say the one involving $Y_1$ and $X_1$. 

These are necessary and sufficient conditions for the susy holomorphic solutions to have 
a constant Gaussian curvature and will be the fundamental equations for our analysis. 

In the following subsection, we take the special case of susy invariant solutions of $G(M,N)$ and show that it solves our problem.


\subsection{Susy invariant solutions}

We now give a sufficient condition for obtaining a constant curvature solution.

Let us first recall that, in the particular case $M = 1$, we have already shown 
that  the susy holomorphic solutions with constant curvature take the form 
(up to gauge transformations)\cite{DHYZ}
\begin{eqnarray}
\omega(x_+,\theta_+) = u(x_+) + i \theta_+ \eta(x_+) \partial_+ u(x_+)\,, 
\label{susyextofhol}
\end{eqnarray}
where $u(x_+)^T = (u_1(x_+), \dots, u_{N-1}(x_+) )$ is the Veronese sequence with
\begin{eqnarray}
u_n(x_+) = \sqrt{\left(\begin{array}{c}
N-1 \\
n \end{array}\right)}\, x_+^n\,, 
\quad n = 0,1,2, \ldots, N-1\,.
\label{revisedefinesusyuxnew}
\end{eqnarray}
Such a solution is called susy invariant  \cite{DHYZ}  since, using Taylor expansion, 
we have $\omega(x_+,\theta_+)=\omega(y_+)$ with $y_+= x_++ i \theta_+ \eta(x_+)$ being a susy translated variable. 

In this section, we 
generalize this result to susy grassmannian $G(M,N)$ $\sigma$-model. 
Indeed, assuming that the susy holomorphic solution is similarly given by
\begin{eqnarray}
W(x_+,\theta_+) = Z(x_+) + i \theta_+ \eta(x_+) \partial_+ Z(x_+),
\label{susyholsol}
\end{eqnarray}
{\it i.e.} $A(x_+)=  \partial_+ Z(x_+)$ in (\ref{MFparasusy}) and keeping 
in mind that the holomorphic solution of the bosonic 
grassmannian $G(M,N)$ $\sigma$-model is written in the MacFarlane 
parametrization \cite{MacFarlane}, we can rewrite (\ref{susyholsol}) as 
\begin{eqnarray}
W = \left(\begin{array}{c}
\mathbb{I}_M \\
K + i \theta_+ \eta \partial_+ K
\end{array}\right).
\label{macfarlane_W}
\end{eqnarray}
Here we have taken into account the susy gauge invariance \cite{HLYZ}.
Since the curvature and metric associated with such a solution are given by (\ref{curvsusy}) 
and (\ref{metricsusy}) respectively, we compute first the determinant of the matrix 
$ W^\dagger W$ which could be written as 
\begin{eqnarray}
W^\dagger W 
= \left(1 + \mathcal{D}\right)\left(\mathbb{I}_M + K^\dagger K\right),
\label{detWW_IK}
\end{eqnarray}
where the differential operator $\mathcal{D}$ is given by 
\begin{eqnarray}
\mathcal{D} = i T_1 \partial_+ 
+ i T_2 \eta^{\dagger} \partial_- 
- T_1 T_2  \partial_+ \partial_- ,
\label{definitionoperatorD}
\end{eqnarray}
using the notation introduced in (\ref{T}).

In order to proceed with the determinant, we use the following lemma which is proven in the Appendix \ref{seclemma1}. 
\begin{lemma}
Let $\mathcal{D}$ be the operator defined in (\ref{definitionoperatorD}) and $B(x_+,x_-)$ 
is an $M\times M$ bosonic matrix. Then we have 
\begin{eqnarray}
\det \left[\left(1 + \mathcal{D} \right) B(x_+,x_-) \right] 
= \left(1 + \mathcal{D} \right) \det\left[B(x_+,x_-)\right].
\label{eqforlemma1}
\end{eqnarray}
\label{lemma1}
\end{lemma}
Replacing $B(x_+,x_-)$ by $\left(\mathbb{I}_M + K^\dagger K\right)$ in the above Lemma, 
we get
\begin{eqnarray}
\det (W^\dagger W) = 
\left(1+\mathcal{D}\right)\det\left(\mathbb{I}_M + K^\dagger K\right).
\label{implicationoflemma}
\end{eqnarray}

Now we can give the following theorem. 
\begin{theorem}
Let us assume that $Z:S^2\rightarrow G(M,N)$, 
is a holomorphic solution of the bosonic 
Euler-Lagrange equations 
associated with a constant Gaussian curvature surface. 
Its susy invariant holomorphic extension (\ref{susyholsol})
is also associated with a constant Gaussian curvature surface of the same curvature. 
\label{theorem1}
\end{theorem}

\begin{proof}
By hypothesis, $Z$ is a holomorphic solution of the bosonic 
model associated with a constant Gaussian curvature surface. It means that \cite{DHZ2013I} there exists an integer $r$ such that 
\begin{equation}
\det \left(Z^\dagger Z\right) = \det\left(\mathbb{I}_M 
+ K^\dagger K\right) = R = 
\left(1 + |x|^2\right)^r\,,
\label{thm1res1}
\end{equation}
for some positive integer $r$ and thus 
$\kappa=\frac{2}{r}$.

In order to get the expression of the metric (\ref{metricsusy}) we first show that 
\begin{equation}
\ln\left[\left(1+\mathcal{D}\right)R\right] = \left(1+\mathcal{D}\right)\ln R.
\label{lnR}
\end{equation}
Using the Taylor expansion of the logarithmic function (\ref{Taylorlog})
and applying it with $x=\frac{1}{R}\mathcal{D} R$, only the first two 
terms of the expansion contribute because $(\mathcal{D}R)^3=0$. We thus get
\begin{eqnarray}
&\ln\left[\left(1+\mathcal{D}\right)R\right] = 
\ln R + \ln\left[1+\frac{1}{R}\mathcal{D} R\right]
= \ln R+ \frac{1}{R}\mathcal{D} R - \frac{1}{2} (\frac{1}{R}\mathcal{D} R)^2\,, \nonumber \\
& = \ln R + i T_1 ( \frac{1}{R}\partial_+ R) 
+ i T_2 (\frac{1}{R}\partial_- R )
-T_1 T_2 \left(\frac{1}{R}\partial_+ 
\partial_- R - \frac{1}{R^2} 
(\partial_+ R) (\partial_- R)\right)\,, \nonumber \\
&= \left(1+\mathcal{D}\right)\ln R.
\label{thm1proof1eq2}
\end{eqnarray}

The next step is to show that (see Appendix \ref{appthm})
\begin{eqnarray}
\tilde{g} &=& \partial_+ \partial_- \ln\left[\left(1+\mathcal{D}\right)R\right] 
= \partial_+ \partial_-  \big(\left(1+\mathcal{D}\right)\ln R \big)\,, \nonumber \\
&=&\left(1+\mathcal{D}+\mathcal{D_\eta} \right)(\partial_+ \partial_-\ln R)\,,
\label{metricthm1proof2}
\end{eqnarray}
where
\begin{eqnarray}
\mathcal{D_\eta} = i  \theta_+ (\partial_+ \eta) 
+ i \theta_- (\partial_- \eta^{\dagger})
- \theta_+ \theta_- \Big((\partial_- \eta^{\dagger}) 
(\partial_+\eta) + (\partial_- \eta^{\dagger}) 
\eta \partial_+ + \eta^{\dagger} (\partial_+ \eta) \partial_- \Big)\,.
\label{rdefmatcaleta}
\end{eqnarray}
Thus the metric becomes
\begin{equation}
\tilde{g}=\left(1+\mathcal{D}+\mathcal{D_\eta} \right)g\,,
\label{gtildefin}
\end{equation}
and (see Appendix \ref{appthm})
\begin{equation}
\ln \tilde{g} = \ln \left[
\left(1 + \mathcal{D} + \mathcal{D_\eta}\right) g \right] 
= (1 + \mathcal{D}) \ln g 
+ i  \theta_+ (\partial_+ \eta) 
+ i \theta_- (\partial_- \eta^{\dagger})\,.
\label{proof324}
\end{equation}
Taking the mixed derivative, we get 
\begin{equation}
\partial_+ \partial_- \ln \tilde{g} = 
\partial_+ \partial_- \Big((1 + \mathcal{D}) \ln g \Big) 
= \left(1 + \mathcal{D} + \mathcal{D_\eta} \right) (\partial_+ \partial_- \ln g),
\end{equation}
a similar result as in (\ref{metricthm1proof2}).
Using the expression of the susy Gaussian curvature (\ref{curvsusy}) and the fact 
that $\partial_+ \partial_- \ln g = -\kappa g$, we finally get
\begin{equation}
- \tilde{\kappa} \tilde{g} = \partial_+ \partial_- \ln \tilde{g} 
= - \kappa \left(1 + \mathcal{D} + \mathcal{D_\eta}\right)  g\,,
\end{equation}
with $\tilde{g}$ given in (\ref{gtildefin}). We conclude that $\tilde{\kappa}=\kappa=\frac{2}{r}$.
\end{proof}

\section{Constant curvature holomorphic solutions of the susy $G(2,4)$ $\sigma$-model 
\label{examples}}

In this section we present all the constant curvature holomorphic 
solutions of the susy $G(2,4)$-$\sigma$-model in a canonical form. 

First for the case of $G(2,N)$, the matrices $M_0, \ M_1, \ M_2, \ M_3$ are $2\times 2$ and the quantities $X_1, \ X_2,\ X_3$ are easily computed from (\ref{detWdagW}). We thus get
\begin{eqnarray}
X_1 &= R^{-1} \left[(M_0)_{11}(M_1)_{22}+(M_0)_{22}(M_1)_{11}-(M_0)_{12}(M_1)_{21}-(M_0)_{21}(M_1)_{12}\right], 
\label{RX1X2X3def1} \\
X_2 &= R^{-1} \left[(M_0)_{11}(M_2)_{22}+(M_0)_{22}(M_2)_{11}-(M_0)_{12}(M_2)_{21}-(M_0)_{21}(M_2)_{12}\right],
\label{RX1X2X3def2}\\
X_3 &= R^{-1} [(M_0)_{11}(M_3)_{22}+(M_0)_{22}(M_3)_{11}+(M_1)_{11}(M_2)_{22}+(M_1)_{22}(M_2)_{11} \nonumber \\
&- (M_0)_{12}(M_3)_{21}-(M_0)_{21}(M_3)_{12}-(M_1)_{12}(M_2)_{21}-(M_1)_{21}(M_2)_{12}].
\label{RX1X2X3def}
\end{eqnarray}
with $M_0, \ M_1,  \ M_2,  \ M_3$ given in (\ref{reviseddefineMV}) and $R \equiv \det Z^\dagger Z = \det M_0$.

In \cite{4solutionsG(24)} it has been shown that, up to a $U(4)$ gauge 
transformation all the constant curvature holomorphic solutions of the 
purely bosonic case are given by 
\begin{equation}
\begin{split}
&Z_1=\left(\begin{array}{cc}
1&0\\
0&1\\
x_+&0\\
0&0\\
\end{array}\right), \  Z_2=\left(\begin{array}{cc}
1&0\\
0&1\\
x_+^2 \cos{2t}&\sqrt{2}x_+\cos{t}\\
\sqrt{2}x_+\sin{t}&0\\
\end{array}\right)\,, \quad t\in \mathbb{R} \\[3ex]
&Z_3=\left(\begin{array}{cc}
1&0\\
0&1\\
\sqrt{3}x_+^2&\sqrt{8/3}x_+ \\
0&\sqrt{1/3}x_+\\
\end{array}\right), \  Z_4=\left(\begin{array}{cc}
1&0\\
0&1\\
2x_+^3&\sqrt{3}x_+^2\\
\sqrt{3}x_+
^2&2x_+\\
\end{array}\right)\,.
\end{split}
\label{Z1234}
\end{equation}

Searching for the constant curvature 
holomorphic solutions of the susy model, we 
generalize them in the following way 
\begin{eqnarray}
W_r(x_+) = Z_r(x_+) + i \theta_+ \eta(x_+) A_r(x_+), \qquad r=1,2,3,4,
\label{extholZi}
\end{eqnarray}
where the different  $Z_r$ are given by (\ref{Z1234}). Our aim is to determine the most general 
matrices $A_r(x_+)$ that satisfy the conditions of having a constant curvature. Using the 
gauge invariance of the susy model \cite{HLYZ}, we take 
\begin{equation}
A_r (x_+) = 
\left(\begin{array}{cc}
0 & 0\\
0 & 0\\
\beta_{11}(x_+) & \beta_{12}(x_+)\\
\beta_{21}(x_+) & \beta_{22}(x_+)\\
\end{array}\right)
= \left( \begin{array}{c}
0 \\
\beta(x_+)
\end{array} \right) \,.
\label{ArReducedgeneral}
\end{equation} 
Since the solutions $ Z_r(x_+) $ are all real functions of $x_+$, we assume that it is also the case for $A_r(x_+)$. 
For each holomorphic solution $W_r(x_+)$ given in (\ref{extholZi}), 
the conditions (\ref{necsufconditions12}) and (\ref{necsufconditions3}) 
have to be satisfied. We investigate each of these cases separately. 
Interestingly for $W_3$ and $W_4$, the only solutions are the susy invariant ones. 
However, it is not true for  $W_1$ and $W_2$.

\subsection{The case of $Z_1$}
\label{caseZ1}
This is the simplest solution of the bosonic $G(2,4)$ model with 
$\det Z_1^{\dagger} Z_1 = \left(1 + |x|^2 \right)$, {\it i.e.;} $r=1$ or $\kappa=2$. 
It is easy to see that the condition (\ref{necsufconditions12}) is trivially 
satisfied for $W_1$ given in (\ref{extholZi}). Hence we are left with the condition (\ref{necsufconditions3}). 
It reads as 
\begin{equation}
|x_+ (\partial_+^2 \beta_{22}) + 2 (\partial_+ \beta_{22})|^2 
+|\partial_+^2 \beta_{22}|^2
+|\partial_+^2 \beta_{12}|^2
+|\partial_+^2\beta_{21}|^2 = 0.
\label{examZ1perfectsqr}
\end{equation}	
Since $\beta_{11}$ does not appear in this equation, it will remain arbitrary. Equation (\ref{examZ1perfectsqr}) 
implies that 
\begin{equation}
\partial_+^2 \beta_{12} = 0,
\quad
\partial_+^2 \beta_{21} = 0,
\quad
\partial_+^2 \beta_{22} = 0,
\quad
x_+ (\partial_+^2 \beta_{22}) + 2(\partial_+ \beta_{22}) = 0, 
\label{examZ1impli}
\end{equation}
which further fix the matrix $A_1$. We thus get constant curvature susy holomorphic solutions of the form
\begin{eqnarray}
W_1 = \left(\begin{array}{cc}
1 & 0\\
0 & 1\\
x_+ & 0\\
0 & 0\end{array}\right) + 
i \theta_+ \eta 
\left(\begin{array}{cc}
0 & 0\\
0 & 0\\
\beta_{11}(x_+) & b_1 x_+ + b_0 \\
c_1 x_+ + c_0 & d_0
\end{array}\right), 
\label{exam1W1final}
\end{eqnarray}
where $b_1$, $b_0$, $c_1$, $c_0$ and $d_0$ are arbitrary constants.
 Notice that when 
$b_0 = b_1 = c_0 = c_1 = d_0 = 0$, we get in particular the susy invariant solution. 
It is clear that we have more solutions than the susy invariant one in this case.


\subsection{The case of $Z_2$}
\label{caseZ2}
We have a family of bosonic solutions, labeled by the parameter $t$: 
\begin{eqnarray}
Z_2 (x_+, t) = 
\left(\begin{array}{c} 
\mathbb{I}_2 \\
K_2(t)
\end{array}\right),
\qquad
K_2 (x_+, t) = 
\left(\begin{array}{cc} 
x_+^2 \cos{2t} & \sqrt{2}x_+ \cos{t} \\
\sqrt{2}x_+ \sin{t} & 0 \\
\end{array}\right).
\label{Z2begin}
\end{eqnarray}
Since
\begin{eqnarray}
\det Z_2^{\dagger} Z_2 = 
\det\left(\mathbb{I}_2 + K_2^{\dagger} K_2 \right) = 
\left(1 + |x|^2 \right)^2,
\label{examdetZ2}
\end{eqnarray}
the associated curvature is $\kappa=1$.

In \cite{4solutionsG(24)}, the parameter $t$ can take any real 
values but due to the properties of the trigonometric functions, 
using a residual gauge invariance, 
we have been able to show (see Appendix C) that 
$t \in [0, \pi[$.

Considering now the corresponding susy holomorphic solution
\begin{eqnarray}
W_2(x_+, \theta_+, t) = Z_2(x_+, t) + i \theta_+ \eta(x_+) A_2(x_+, t),
\label{examZ2susysol}
\end{eqnarray}
where  $A_2(x_+, t)$ takes the form (\ref{ArReducedgeneral}), 
the conditions (\ref{necsufconditions12}) and (\ref{necsufconditions3}) 
have to be satisfied in order to get a family of constant curvature solutions. 

Introducing $W_2$ given in (\ref{examZ2susysol}) into 
(\ref{necsufconditions12}), we get two different cases: 
\begin{enumerate}
\item The first case corresponds to $\cos 2t \neq 0$. Condition (\ref{necsufconditions12}) 
implies 
\begin{eqnarray}
\beta_{11}(x_+, t) = 
x_+ \big(\sqrt{2} \cos{t} \beta_{12}(x_+, t) - \sqrt{2}
\sin{t} \beta_{21}(x_+, t) + x_+ \sin{2t} \beta_{22}(x_+, t) \big).
\label{exambeta11Z2}
\end{eqnarray}
So we have only one condition (\ref{necsufconditions3}) to 
resolve three unknown functions. Interestingly, starting with 
a polynomial form in $x_+$ of the unknown functions we get a 
pattern. Indeed, we find that 
\begin{eqnarray}
\beta_{12}(x_+, t) &=& c_0 + c_1 x_+ + F(x_+)\,, \label{examgenB12} \\
\beta_{21}(x_+, t) &=& \big(c_0 + F(x_+) \big) \tan t + a_1 x_+\,, \label{examgenB21} \\
\beta_{22}(x_+, t) &=& \frac{\cos t}{\sqrt{2}} \big(a_1 - c_1 \tan t \big)\,, \label{examgenB22} 
\end{eqnarray}
where $a_1, c_0$ and $c_1$ are constants, solve our problem. Thus the matrix $\beta(x_+, t)$ 
takes the form 
\begin{eqnarray}
\beta(x_+, t) &\!\! = \!\!& 
\frac{(c_0 + F(x_+))}{\sqrt{2} \cos t} 
\left(\begin{array}{cc} 
2 x_+ \cos 2t & \sqrt{2} \cos t  \\
\sqrt{2} \sin t & 0 \\
\end{array}\right)
+ a_1\!
\left(\begin{array}{cc} 
-\sqrt{2} x_+^2 \sin^3 t & 0  \\
x_+ & \frac{1}{\sqrt{2}}\cos t \\
\end{array}\right) \nonumber \\
&\,& +\, c_1\!
\left(\begin{array}{cc} 
\sqrt{2} x_+^2 \cos^3 t & x_+  \\
0 & -\frac{1}{\sqrt{2}}\sin t \\
\end{array}\right).
\label{examB2solution}
\end{eqnarray}
The susy invariant solution is obtained 
when $a_1=c_1=0$. Again the case $W_2$ gives other solutions to our problem than the susy invariant ones.
\item The second case corresponds to  $\cos 2t = 0$ or $t = \frac{\pi}{4}$ (the case $t= \frac{3\pi}{4}$ 
is gauge equivalent) so that 
\begin{eqnarray}
K_2 (x_+, \frac{\pi}{4}) = 
\left(\begin{array}{cc} 
0 & x_+  \\
x_+ & 0 \\
\end{array}\right).
\label{examK2spc1}
\end{eqnarray}
Since $K_2 (x_+, \frac{\pi}{4})$ is  
symmetric, we assume that the matrix $\beta(x_+)$ is also symmetric, {\it i.e.}
\begin{eqnarray}
\beta_{21}(x_+) = \beta_{12}(x_+)\,.
\label{examassumption}
\end{eqnarray}
These quantities will remain arbitrary since the condition (\ref{necsufconditions3}) depends 
only on $\beta_{11}$ and $\beta_{22}$ and the susy invariant solutions will be obtained when $\beta_{11}=\beta_{22}=0$.

The condition (\ref{necsufconditions3}) may be written 
as follows, taking in particular $x_+ = x_- = x$:
\begin{eqnarray}
&(1 + x^2)^2 \Big(4 (x^2-1) \big((\beta_{11}^{\prime})^2+ (\beta_{22}^{\prime})^2\big) 
+  (1 + x^2)^2 \big((\beta_{11}^{\prime\prime})^2+ (\beta_{22}^{\prime\prime})^2\big)\Big) 
\nonumber \\
& - 8 x (1 + x^2) (x^2 - 2)\big(\beta_{11} \beta_{11}^{\prime} + 
\beta_{22} \beta_{22}^{\prime}\big) 
+ 4 x^2    (1 + x^2)^2 \big(\beta_{11} \beta_{11}^{\prime\prime} + 
\beta_{22} \beta_{22}^{\prime\prime}\big) 
\nonumber \\
& + 4 (1- 4x^2 + x^4) (\beta_{11}^2+\beta_{22}^2)
- 4 x (1 + x^2)^3 \big(\beta_{11}^{\prime} \beta_{11}^{\prime\prime} + 
\beta_{22}^{\prime} \beta_{22}^{\prime\prime}\big) = 0\,.
\label{examwithxisequal}
\end{eqnarray}
Let us first mention the invariance of this equation 
with respect to the exchange 
$\beta_{11}\leftrightarrow\beta_{22}$. 

After some trials we first get a solution choosing 
\begin{eqnarray}
\beta_{22}(x) = x \beta_{11}(x)\,.
\label{examtrials1}
\end{eqnarray}
Condition (\ref{examwithxisequal}) thus becomes very 
simple 
\begin{eqnarray}
(1 + x^2)^5 (\beta_{11}^{\prime\prime}(x))^2 = 0\,,
\label{examtrialverysimp}
\end{eqnarray}
which implies that 
\begin{eqnarray}
\beta_{11}(x) = a_0 + d_2 x\,, \qquad  
\beta_{22}(x) = x (a_0 + d_2 x)\,.
\label{examassumimplies}
\end{eqnarray}

Using this observation, we assume that $\beta_{11}(x)$ 
and $\beta_{22}(x)$ are real polynomial 
in  $x$. We can 
easily show that they must be at most 
of degree $2$. If we take 
\begin{eqnarray}
\beta_{11}(x) = a_2 x^2 + a_1 x + a_0\,, \qquad  
\beta_{22}(x) = d_2 x^2 + d_1 x + d_0\,,
\label{examatrials2}
\end{eqnarray}
and identify the coefficients of different powers of $x$ 
in (\ref{examwithxisequal}), we get three independent 
equations for the parameters $a_i$ and $d_i$,
\begin{eqnarray}
&a_0^2-a_1^2+a_2^2+d_0^2-d_1^2+d_2^2 = 0\,, 
\qquad  
a_0 a_2+d_0 d_2 = 0\,,
\nonumber \\
& a_0 a_1 - a_1 a_2 + d_1 (d_0 - d_2) = 0\,.
\label{examcondonparameters}
\end{eqnarray}

Let us first assume that $a_0\neq 0$, we then get
\begin{eqnarray}
\beta_{11}(x) = a_0 + (d_2 - d_0) x - \frac{d_0 d_2}{a_0} x^2\,, \nonumber \\  
\beta_{22}(x) = d_0 + (a_0 + \frac{d_0 d_2}{a_0}) x + d_2 x^2\,,
\label{examatrials3}
\end{eqnarray}
where $a_0$, $d_0$ and $d_2$ remain arbitrary real parameters. Clearly the solution (\ref{examassumimplies}) is 
obtained when $d_0 =0$. 

Now, consider $a_0 = 0$. We then get different subcases ($\epsilon=\pm 1$)
\begin{itemize}
 \item $d_2\neq 0$, $d_0 = 0$ $\Longrightarrow 
\left \{\begin{array}{c} 
\beta_{11}(x) = a_2 x^2 + \epsilon d_2 x\,, \\
\beta_{22}(x) = d_2 x^2 - \epsilon a_2 x\,, \\
\end{array}
\right.$
 \item $d_0\neq 0$, $d_2 = 0$ $\Longrightarrow 
\left \{\begin{array}{c} 
\beta_{11}(x) = a_2 x^2 + \epsilon d_0 x\,, \\
\beta_{22}(x) = \epsilon a_2 x + d_0\,. \\
\end{array}
\right.$
\end{itemize}
\end{enumerate}

\subsection{The case of $Z_3$}
\label{caseZ3}
In this case we have 
$\det Z_3^{\dagger} Z_3 = \left(1 + |x|^2 \right)^3$, {\it i.e.;} $r=3$ and $\kappa=\frac23$. 
With the solution $W_3$ as in (\ref{extholZi}), the 
condition (\ref{necsufconditions12}) becomes a third degree polynomial 
in $x_-$. Equating the coefficients of different powers of $x_-$ to 
zero we obtain the following equations :
\begin{eqnarray}
&&2 x_+^3 \left(\sqrt{2} \beta_{12}'' + 5 \beta_{22}'' \right) 
- x_+^2 \left(3 \beta_{11}'' + 8 \sqrt{2} \beta_{12}' 
+ 6 \sqrt{2} \beta_{21}'' + 40 \beta_{22}'\right) \nonumber \\
&&+ 6 x_+ 
\left(3 \beta_{11}' + 2 \sqrt{2} \beta_{12} + 6 \sqrt{2} \beta_{21}' + 10 \beta_{22}\right) 
-36 \beta_{11} -72 \sqrt{2} \beta_{21} = 0\,,  
\label{Z34conditions1rev} \\
&&x_+^2 \left(-\beta_{11}'' + 8 \sqrt{2} \beta_{12}' - 4 \sqrt{2} 
\beta_{21}'' + 6 x_+ \beta_{22}'' + 4 \beta_{22}'\right) \nonumber \\
&&- x_+ \left(4 \beta_{11}' 
+ 24 \sqrt{2} \beta_{12} - 8 \sqrt{2} \beta_{21}' 
+ 48 \beta_{22}\right) 
+ 28 \beta_{11} + 16 \sqrt{2} \beta_{21} = 0\,,
\label{Z34conditions2rev} \\
&&x_+ \left(\beta_{11}'' - 2 \sqrt{2} x_+ \beta_{12}'' + 8 \sqrt{2} \beta_{12}' 
- 2 \sqrt{2} \beta_{21}'' + 2 x_+ \beta_{22}'' + 16 \beta_{22}'\right) \nonumber \\
&&- 10 \beta_{11}' + 12 \sqrt{2} \beta_{12} - 4 \sqrt{2} \beta_{21}' 
+ 12 \beta_{22} = 0\,,
\label{Z34conditions3rev} \\
\nonumber \\
&&-3 \beta_{11}'' + 4 \sqrt{2} x_+ \beta_{12}'' + 8 \sqrt{2} \beta_{12}' 
+ 2 x_+ \beta_{22}'' + 4 \beta_{22}' = 0\,.
\label{Z34conditions4rev}
\end{eqnarray}

We first solve (\ref{Z34conditions4rev}) for $\beta_{11}$ and get 
\begin{eqnarray}
\beta_{11}(x_+)=\frac{4 \sqrt{2} x_+}{3} \beta_{12}(x_+) 
+ \frac{2 x_+}{3} \beta_{22}(x_+) + c_1 x_+ + c_2\,,
\label{Z34integratedrev}
\end{eqnarray}
where $c_1$ and $c_2$ are integration constants. Upon 
introducing (\ref{Z34integratedrev}) into some linear 
combinations of (\ref{Z34conditions1rev}), (\ref{Z34conditions2rev}) 
and (\ref{Z34conditions3rev}) we obtain 
\begin{eqnarray}
\beta_{22}(x_+) = \frac{\sqrt{2}}{4} \beta_{12}(x_+) 
+ \frac{3\sqrt{2}}{4 x_+} \beta_{21}(x_+) + \frac{3}{4} c_1 + \frac{3}{4 x_+} c_2\,.
\label{Z34integratedrev2}
\end{eqnarray} 
In order to satisfy all the 
equations (\ref{Z34conditions1rev})-(\ref{Z34conditions3rev}), the integration constants $c_1$ 
and $c_2$ must vanish. Hence, 
we can give the final form of $\beta_{11}$ and $\beta_{22}$ as 
\begin{eqnarray}
\beta_{11}(x_+)=\frac{3 x_+}{\sqrt{2}} \beta_{12}(x_+) 
+ \frac{1}{\sqrt{2}} \beta_{21}(x_+),
\label{exam3beta11Z3}
\end{eqnarray}
\begin{eqnarray}
\beta_{22}(x_+) = \frac{\sqrt{2}}{4} \beta_{12}(x_+) 
+ \frac{3\sqrt{2}}{4 x_+} \beta_{21}(x_+).
\label{exam3beta22Z3}
\end{eqnarray} 
Introducing (\ref{exam3beta11Z3}) and (\ref{exam3beta22Z3}) into the condition (\ref{necsufconditions3}) we obtain 
\begin{eqnarray}
\partial_+ \partial_- 
\left(
\frac{\big|\big(1 + 3 |x|^2\big) \beta_{21} 
- x_+ \big(1 + |x|^2\big) \partial_+ \beta_{21}\big|^2}{|x|^4 \big(1 + |x|^2\big)^2}
\right) = 0\,,
\label{exam3towardslastcon}
\end{eqnarray}
or equivalently
\begin{eqnarray}
\left(
\frac{\big|\big(1 + 3 |x|^2\big) \beta_{21} 
- x_+ \big(1 + |x|^2\big) \partial_+ \beta_{21}\big|^2}{|x|^4 \big(1 + |x|^2\big)^2}
\right) = f(x_+) + g(x_-)\,,
\label{exam3getteingtheresult}
\end{eqnarray}
for arbitrary functions $f$ and $g$ of given variables.
Requiring it to be satisfied when $x_+ = 0$ and 
$x_- = 0$ separately we obtain 
\begin{eqnarray}
\beta_{21} (x_+) = \gamma_1 x_+\,,
\label{exam3solforr}
\end{eqnarray}
where $\gamma_1$ is an arbitrary constant. Upon 
introducing it into (\ref{exam3getteingtheresult}) 
we get 
\begin{eqnarray}
f(x_+) + g(x_-) = \frac{4 |x|^2 \gamma_1^2}{(1 + |x|^2)^2}\,,
\label{exam3resultfinal}
\end{eqnarray}
which immediately implies that $\gamma_1 = 0$ and hence $\beta_{21} = 0$. 

The necessary and sufficient conditions 
(\ref{necsufconditions12}) and (\ref{necsufconditions3}) 
are thus satisfied and finally the
constant curvature holormorphic solution $W_3$ is 
given by the form
\begin{equation}
W_3= Z_3+ i\theta_+ \sqrt{3} \eta(x_+)  \beta_{22}(x_+) \partial_+ Z_3.
\label{exam3final}
\end{equation}
Hence in this case we have obtained the susy invariant solution as the unique constant curvature holomorphic solution.


\subsection{The case of $Z_4$}
\label{caseZ4}
In this case we have 
$\det Z_4^{\dagger} Z_4 = \left(1 + |x|^2 \right)^4$, {\it i.e.;} $r=4$ and $\kappa=\frac12$.

Again the condition (\ref{necsufconditions12}) becomes a third degree 
polynomial in $x_-$ after introducing the solution $W_4$ given in 
(\ref{extholZi}). Similarly as what we did with $W_3$, we equate the coefficients of different powers 
of $x_-$ to zero and now get  
\begin{eqnarray}
\beta_{11}(x_+) = 3 x_+^2 \beta_{22}(x_+),
\label{exam4beta11Z4}
\end{eqnarray} 
\begin{eqnarray}
\beta_{21}(x_+)= - \beta_{12}(x_+) 
+ 2 \sqrt{3} x_+ \beta_{22}(x_+).
\label{exam4beta21Z4}
\end{eqnarray}

We are left with the last condition (\ref{necsufconditions3}). 
Introducing (\ref{exam4beta11Z4}) and (\ref{exam4beta21Z4}) into this 
last condition we find that $\beta_{21}(x_+) = \beta_{12}(x_+)$.
Finally,  the constant curvature holomorphic solution $W_4$ is 
given as
\begin{equation}
W_4= Z_4+ i\theta_+\frac12 \eta(x_+)  \beta_{22}(x_+) \partial_+ Z_4.
\label{exam4final}
\end{equation}
Again, we have obtained in this case the susy invariant solution as the unique constant curvature holomorphic solution.

\section{About the Pl{\"u}cker embedding of $G(2,4)$ into $\field{C}P^{5}$
\label{casecp5}}
It is well-known \cite{Calabi, Griffiths, Yang} for the bosonic model that Pl{\"u}cker embedding of $G(2,N)$ into $\field{C}P^{\frac{N(N-1)}{2}-1}$ is obtained by introducing the map $\Phi_N : G(2,N) \to \field{C}P^{\frac{N(N-1)}{2}-1}$.

In our case, we get explicitly the map  $\Phi_4 : G(2,4) \to \field{C}P^{5}$, on the form
\begin{equation}
\Phi_4 (Z)= (1 ,\ -z_{31}, \ z_{32}, \ -z_{41}, \ z_{42} ,\ z_{31} z_{42}- z_{32} z_{41})^T,
\label{plucker}
\end{equation}
when
\begin{equation}
Z= 
\left(\begin{array}{cc}
1 & 0\\
0 & 1\\
z_{31}& z_{32}\\
z_{41} & z_{42}\\
\end{array}\right).
\label{Zgeneral}
\end{equation} 
It is thus easy to show, up to gauge equivalence $\hat{\Phi}_4 (Z_r)=V_r \Phi_4 (Z_r)$ ($V_r \in U(6)$ is a constant matrix, explicitly given in Appendix D for $r = 1, 2, 3, 4$) that we get
\begin{equation}
\hat{\Phi}_4 (Z_1)= (1, \ x_+ ,\ 0, \ 0, \ 0, \ 0)^T,
\label{pluckerz1}
\end{equation}
\begin{equation}
\hat{\Phi}_4 (Z_2)= (1, \ \sqrt{2} x_+ ,\ x_+^2, \ 0, \ 0, \ 0)^T,
\label{pluckerz2}
\end{equation}
\begin{equation}
\hat{\Phi}_4 (Z_3)= (1, \ \sqrt{3} x_+, \ \sqrt{3} x_+ ^2,\ x_+^3, \ 0, \ 0)^T,
\label{pluckerz3}
\end{equation}
\begin{equation}
\hat{\Phi}_4 (Z_4)= (1, \ 2 x_+ ,\ \sqrt{6} x_+^2, \ 2 x_+^3, \ x_+ ^4,\ 0)^T.
\label{pluckerz4}
\end{equation}

For the bosonic case, such a correspondence has helped \cite{DHZ2013I} constructing the holomorphic solutions with constant curvature of $G(2,N)$ from the Veronese curves embedded in
 $\field{C}P^{\frac{N(N-1)}{2}-1}$.
 
The Veronese curve in $\field{C}P^{5}$ 
 \begin{equation}
(1, \ \sqrt{5}  x_+, \ \sqrt{10} x_+^2, \ \sqrt{10}  x_+^3, \ \sqrt{5} x_+ ^4,\ x_+^5)^T
\label{plucker5}
\end{equation} 
does not give rise to a solution of  $G(2,4)$. Indeed, the constraint on the Pl{\"u}cker coordinates is not satisfied \cite{DHZ2013I}.

For the susy case, we see that the arbitrariness in the possible choices of the fermionic contribution $A(x_+)$ may thus come from some arbitrariness in the corresponding solutions of $\field{C}P^{5}$. 
A detailed discussion of such a correspondence is out of the scope of this paper.




\section{Conclusions and final comments \label{conc}}
In this article we give some criteria for having constant curvature 
holomorphic solutions of the susy grassmannian $G(M,N)$ 
$\sigma$-model. With the help of the susy gauge invariance of the 
model we first show that the susy holomorphic solution 
given in (\ref{susyholsol}) ({\it i.e.}; generalisation of 
bosonic holomorphic solution) leads to a constant curvature 
surface. This kind of a solution is called a susy invariant one, 
in analogy with the discussion given in \cite{DHYZ}. 

Then we restrict ourselves to the susy $G(2,N)$ $\sigma$-model 
and give the necessary and sufficient conditions to get such solutions.
The case of  $G(2,4)$ is studied in detail taking into account the classification of bosonic solutions \cite{4solutionsG(24)}. 

The existence of an embedding of $G(2,4)$ into $\field{C}P^{5}$ shows a connection between the corresponding solutions. It will be relevant when we consider higher dimensional models. Indeed the case of the susy $G(2,5)$ 
$\sigma$-model could be treated taking into account the results for the bosonic case \cite{solutionsG(25)} and its relation to $\field{C}P^{9}$. Some results for the bosonic $G(2,N)$ \cite{solutionsG(2n)} will thus be used to study similar solutions of the susy $G(2,N)$ $\sigma$-model.





\section*{Acknowledgments}
\.{I}Y thanks the Centre de Recherches Math\'{e}matiques, Universit\'{e} de Montr\'{e}al, 
where part of this work was done, for the kind hospitality. VH acknowledges the support 
of research grants from NSERC of Canada. Several discussions with W J Zakrzewski have been very fruiful.

\appendix
\section{Proof of Lemma \ref{lemma1} \label{seclemma1}}
Here we will prove Lemma \ref{lemma1}. 
\begin{proof}
Let us first show that the Lemma \ref{lemma1} holds for $M = 2$ and then generalize it to general $M$. 
For $M = 2$, we have 
\begin{eqnarray}
(1+\mathcal{D})B=\left(\begin{array}{cc}
\left(1+\mathcal{D}\right)b_{11} & \left(1+\mathcal{D}\right)b_{12}\\
\left(1+\mathcal{D}\right)b_{21} & \left(1+\mathcal{D}\right)b_{22}\\
\end{array}\right),
\label{WWN=2provelemma}
\end{eqnarray}	
whose determinant is
\begin{eqnarray}
\det\left[\left(1+\mathcal{D}\right)B\right] 
= (1+\mathcal{D})b_{11}(1+D)b_{22}-(1+\mathcal{D})b_{12}(1+D)b_{21}.
\label{detWWN=2provelemma2}
\end{eqnarray}
It is enough to show that (\ref{detWWN=2provelemma2}) can be expressed 
as $\left(1+\mathcal{D}\right)\left(b_{11}b_{22}-b_{12}b_{21}\right)$. Let us 
consider the following two equations which has to be verified: 
\begin{eqnarray}
\left(1+\mathcal{D}\right)b_{11}\left(1+\mathcal{D}\right)b_{22} &=& \left(1+\mathcal{D}\right)
\left(b_{11}b_{22}\right),
\label{twoterms1} \\
(1+\mathcal{D})b_{12}(1+\mathcal{D})b_{21} &=& \left(1+\mathcal{D}\right)(b_{12}b_{21}).
\label{twoterms2}
\end{eqnarray}
We concentrate our calculations on (\ref{twoterms1}). The result for (\ref{twoterms2}) will then follow 
immediately. In order to simplify the calculations we separate the differential 
operator $\mathcal{D}$ into first- and second-order parts as 
\begin{eqnarray}
\mathcal{D}_1 = i\theta_+ \eta \partial_+ + i \theta_- \eta^{\dagger} \partial_-, 
\qquad 
\mathcal{D}_2 = - \theta_+ \theta_- |\eta|^2 \partial_+\partial_-.
\label{separate}
\end{eqnarray}
Since $\mathcal{D}_1$ is a first-order operator 
and $\mathcal{D}_2$ is a second-order operator, it is not 
difficult to verify (\ref{twoterms1}) by considering the 
properties of grassmann variables and the following identity 
\begin{eqnarray}
\mathcal{D}_2 (b_{11}b_{22}) = 
b_{11}(\mathcal{D}_2 b_{22}) 
+ (\mathcal{D}_2 b_{11})b_{22} 
+ (\mathcal{D}_1 b_{11})(\mathcal{D}_1 b_{22}).
\label{identityprovinglemma}
\end{eqnarray}
Hence the Lemma \ref{lemma1} is shown to be true for $M = 2$. 
For general $M > 2$, $(1+\mathcal{D})B$ will be an $M\times M$ 
matrix whose determinant can be expressed as  
\begin{eqnarray}
\det \left[(1+\mathcal{D})B\right]
= 
\sum(-1)^{sg(\nu_i)} (1+\mathcal{D})b_{1\nu_1}(1+\mathcal{D})b_{2\nu_2}...(1+\mathcal{D})b_{N-1\nu_{N-1}},
\label{detgeneral}
\end{eqnarray}
where the sum is over the permutations $\nu_i$. 
By using (\ref{twoterms1}) it is clear that the terms of this sum can be rewritten as
\begin{eqnarray}
&(1+\mathcal{D})b_{1\nu_1} (1+\mathcal{D})b_{2\nu_2}(1+\mathcal{D})b_{3\nu_3}...
(1+\mathcal{D})b_{{N-1}\nu_{N-1}} \nonumber \\ 
& =(1+\mathcal{D})(b_{1\nu_{1}}b_{2\nu_{2}})(1+\mathcal{D})b_{3\nu_3}...
(1+\mathcal{D})b_{{N-1}\nu_{N-1}},
\label{detgeneralreorganise}
\end{eqnarray}
and following the same strategy, they are equal to 
\begin{eqnarray}
&(1+\mathcal{D})(b_{1\nu_{1}}b_{2\nu_{2}}) (1+\mathcal{D})b_{3\nu_3}(1+\mathcal{D})b_{4\nu_4}...
(1+\mathcal{D})b_{{N-1}\nu_{N-1}} \nonumber \\
& =(1+\mathcal{D})(b_{1\nu_{1}}b_{2\nu_{2}}b_{3\nu_3})(1+\mathcal{D})b_{4\nu_4}...
(1+\mathcal{D})b_{N-1\nu_{N-1}}.
\label{detgeneralreorganiseinter1}
\end{eqnarray}		
By iteration, one obtains 
\begin{eqnarray}
(1+\mathcal{D})b_{1\nu_1}(1+\mathcal{D})b_{2\nu_2}...
(1+\mathcal{D})b_{N-1\nu_{N-1}}=(1+\mathcal{D})(b_{1\nu_1}b_{2\nu_2}b_{3\nu_3}...
b_{{N-1}\nu_{N-1}}).
\label{detgeneralreorganiseinter2}
\end{eqnarray}
By applying this argument to all of the terms in the sum (\ref{detgeneral}), 
it is clear that for any value of $M$, we have 
\begin{eqnarray}
\det\left[(1+\mathcal{D})B(x_+,x_-)\right]=(1+\mathcal{D})\det\left[B(x_+,x_-)\right].
\label{lemmaresult}
\end{eqnarray}
\end{proof}


\section{Equations (\ref{metricthm1proof2}) and (\ref{proof324})\label{appthm}}
In this appendix we give some explicit calculations, which are needed in the 
proof of the theorem \ref{theorem1}. In particular, we derive 
(\ref{metricthm1proof2}) and (\ref{proof324}). In (\ref{metricthm1proof2}) we have 
\begin{eqnarray}
\tilde{g} = \partial_+ \partial_- \big(\left(1+\mathcal{D}\right)\ln R \big) 
= \partial_+ \partial_-\ln R + \partial_+ \partial_- \big(\mathcal{D} \ln R \big) \,.
\label{rappeq1}
\end{eqnarray}
Since $\eta$ and $\eta^{\dagger}$ in the definition of operator $\mathcal{D}$ (\ref{definitionoperatorD}) 
are functions of $x_+$ and $x_-$, respectively, we can develop the last term in (\ref{rappeq1}) as 
\begin{eqnarray}
\partial_+ \partial_- \big(\mathcal{D} \ln R \big) &=& 
\partial_+ \partial_- \big(i \theta_+ \eta \partial_+ \ln R 
+ i \theta_- \eta^{\dagger} \partial_- \ln R 
- \theta_+ \theta_-  \eta^{\dagger}  \eta  \partial_+ \partial_- \ln R 
\big)\,, \nonumber \\
&=& \partial_+ 
\Big(i \theta_+ \eta \partial_+ \partial_- \ln R 
+ i \theta_- \big(\partial_- \eta^{\dagger} \big) \partial_- \ln R 
+ i \theta_- \eta^{\dagger} \partial_- \partial_- \ln R \nonumber \\
&\,& +  \theta_+ \theta_-  \big(\partial_- \eta^{\dagger}\big) \eta  \partial_+ \partial_- \ln R 
- \theta_+ \theta_-  \eta^{\dagger}  \eta \partial_- \partial_+ \partial_- \ln R 
\Big)\,. \nonumber \\ 
\end{eqnarray}
Deriving again with respect to $\partial_+$, we finally get 
\begin{eqnarray}
\partial_+ \partial_- \big(\mathcal{D} \ln R \big) &=& \Bigg(\mathcal{D} 
+ i \theta_+ \big(\partial_+ \eta\big)  
+ i \theta_- \big(\partial_- \eta^{\dagger} \big) \nonumber \\ 
&\,& 
- \theta_+ \theta_-  \Big(\big(\partial_- \eta^{\dagger}\big) \big(\partial_+ \eta\big) 
+  \big(\partial_- \eta^{\dagger}\big)  \eta \partial_+ 
+ \eta^{\dagger} \big(\partial_+ \eta\big) \partial_- \Big)
\Bigg) \partial_+ \partial_- \ln R \,, \nonumber \\ 
&=& \Big(\mathcal{D}+\mathcal{D_\eta}
\Big) \partial_+ \partial_- \ln R \,,
\label{rappeq2}
\end{eqnarray}
where we define $\mathcal{D_\eta}$ as in (\ref{rdefmatcaleta}) 
and hence (\ref{metricthm1proof2}) holds. 

Let us now show that (\ref{proof324}) holds. Using the Taylor expansion of the 
logarithmic function (\ref{Taylorlog}), we have 
\begin{eqnarray}
\ln \tilde{g} &=& \ln \left[
\left(1 + \mathcal{D} + \mathcal{D_\eta}\right) g \right]\,, \nonumber \\
&=& \ln g + \frac{1}{g} \big(\mathcal{D} + \mathcal{D_\eta}\big) g 
- \frac{1}{2} \Big(\frac{1}{g} \big(\mathcal{D} + \mathcal{D_\eta}\big) g\Big)^2\,.
\label{rappendixproof324}
\end{eqnarray} 
Expressions $\left(\mathcal{D} + \mathcal{D_\eta}\right) g$ 
and $\big(\left(\mathcal{D} + \mathcal{D_\eta}\right) g\big)^2$ 
in (\ref{rappendixproof324}) can be expressed as 
\begin{eqnarray}
\left(\mathcal{D} + \mathcal{D_\eta}\right) g &=& 
\Bigg( i \theta_+ \big((\partial_+ \eta) + \eta \partial_+\big) 
+ i \theta_- \big((\partial_- \eta^{\dagger}) + \eta^{\dagger} \partial_- \big)  \nonumber \\
&-& \theta_+ \theta_- 
\Big((\partial_- \eta^{\dagger})(\partial_+ \eta) 
+ (\partial_- \eta^{\dagger})\eta\partial_+ 
+ \eta^{\dagger} (\partial_+ \eta) \partial_- 
+ \eta^{\dagger} \eta \partial_+ \partial_- \Big)\Bigg) g,
\label{rappendiixeqinlog1}
\end{eqnarray}
and
\begin{eqnarray}
\big(\left(\mathcal{D} + \mathcal{D_\eta}\right) g\big)^2 &=& 
- 2\, \theta_+ \theta_- \Big( 
(\partial_- \eta^{\dagger})(\partial_+ \eta) g^2
+ \eta^{\dagger} (\partial_+ \eta) (\partial_- g) g \nonumber \\
&+& (\partial_- \eta^{\dagger})\eta g (\partial_+ g) 
+ \eta^{\dagger} \eta (\partial_+ g) (\partial_- g) \Big)\,.
\label{rappendiixeqinlog2}
\end{eqnarray}

Upon introducing (\ref{rappendiixeqinlog1}) and (\ref{rappendiixeqinlog2}) 
into (\ref{rappendixproof324}) and making necessary cancellations we arrive 
at (\ref{proof324}).


\section{Residual gauge invariance for $Z_2$}
Here we want to prove that the bosonic holomorphic solutions $Z_2(x_+,t)$ may be considered in the interval 
$t \in [0, \pi[$. For obtaining the admissible gauge transformations $V$, we refer \cite{HLYZ}. In the case of $Z_2$ 
we have 
\begin{eqnarray}
K = K_2 =  \left(\begin{array}{cc} 
x_+^2 \cos 2t & \sqrt{2} x_+ \cos t  \\
\sqrt{2} x_+ \sin t & 0 \\
\end{array}\right)\,.
\label{examZ2matK}
\end{eqnarray}

First notice that since $V$ is a constant matrix and 
$\lim \limits_{x_+ \to 0} K_2 = 0$, we have 
\begin{eqnarray}
\lim_{x_+ \to 0} K_{2G} = 0\,.
\label{examlimit}
\end{eqnarray}
Thus $V_{21} = 0$ and by unitarity $V_{12} = 0$. 
Finally, $K_{2G}$ and $K_{2}$ are related as 
\begin{eqnarray}
K_{2G} = V_{22} K_2 V_{11}^{\dagger}\,.
\label{examrelationgauge}
\end{eqnarray}
Moreover, since the entries of $K_2$ 
are real functions of $x_+$, the matrices $V_{11}$ and 
$V_{22}$ are real and orthogonal. In particular, this means 
that 
\begin{eqnarray}
\det K_{2G} = \pm \det K_2 = \pm x_+^2 \sin 2t \,.
\label{examrealorthmean}
\end{eqnarray}
Since $K_{2G}$ has the same form as $K_2$, it must have 
the same determinant as $K_2$ up to a sign. 

Assuming that $t_0$ is fixed in $K_2$, the only admissible 
values of $t$ in $K_{2G}$ are 
\begin{eqnarray}
t = \pm t_0 + k \frac{\pi}{2}\,.
\label{examadmissiblevalK}
\end{eqnarray}
It can easily be shown that $V_{11}$ and 
$V_{22}$ can be fixed, such that 
\begin{eqnarray}
K_2 (\pm t_0 + \pi) \simeq K_2 (t_0)\,,
\label{examfixingV11V221}
\end{eqnarray}
and 
\begin{eqnarray}
K_2 (\pm t_0 + \frac{3\pi}{2}) \simeq K_2 (t_0 + \frac{\pi}{2})\,.
\label{examfixingV11V222}
\end{eqnarray}
$K_2 (t_0)$ and $K_2 (t_0 + \frac{\pi}{2})$ are not gauge equivalent and we have reduced the interval of 
values of the parameter $t$ between $[0, \pi[$.


\section{The transformation matrices $V_r$ of Section \ref{casecp5} \label{transmatapp}}
Here, we give the explicit form of the transformation matrices $V_r$ that we used 
in Section \ref{casecp5}. 
\begin{eqnarray}
&& V_1 = \left(
\begin{array}{cccccc}
 1 & 0 & 0 & 0 & 0 & 0 \\
 0 & -1 & 0 & 0 & 0 & 0 \\
 0 & 0 & 1 & 0 & 0 & 0 \\
 0 & 0 & 0 & 1 & 0 & 0 \\
 0 & 0 & 0 & 0 & 1 & 0 \\
 0 & 0 & 0 & 0 & 0 & 1 \\
\end{array}
\right)\!, 
V_2 = \left(
\begin{array}{cccccc}
 1 & 0 & 0 & 0 & 0 & 0 \\
 0 & 0 & \cos t & -\sin t & 0 & 0 \\
 0 & -\cos 2 t & 0 & 0 & 0 & -\sin 2 t \\
 0 & 0 & \sin t & \cos t & 0 & 0 \\
 0 & 0 & 0 & 0 & 1 & 0 \\
 0 & -\sin 2 t & 0 & 0 & 0 & \cos 2 t \\
\end{array}
\right)\!, 
\nonumber \\
\\
&& V_3 = \left(
\begin{array}{cccccc}
 1 & 0 & 0 & 0 & 0 & 0 \\
 0 & 0 & \frac{2 \sqrt{2}}{3} & 0 & \frac{1}{3} & 0 \\
 0 & -1 & 0 & 0 & 0 & 0 \\
 0 & 0 & 0 & 0 & 0 & 1 \\
 0 & 0 & 0 & 1 & 0 & 0 \\
 0 & 0 & \frac{1}{3} & 0 & -\frac{2 \sqrt{2}}{3} & 0 \\
\end{array}
\right)\!, 
V_4 = \left(
\begin{array}{cccccc}
 1 & 0 & 0 & 0 & 0 & 0 \\
 0 & 0 & 0 & 0 & 1 & 0 \\
 0 & 0 & \frac{1}{\sqrt{2}} & -\frac{1}{\sqrt{2}} & 0 & 0 \\
 0 & -1 & 0 & 0 & 0 & 0 \\
 0 & 0 & 0 & 0 & 0 & 1 \\
 0 & 0 & \frac{1}{\sqrt{2}} & \frac{1}{\sqrt{2}} & 0 & 0 \\
\end{array}
\right)\,.
\nonumber
\label{tranmatappendix}
\end{eqnarray}



\begin{thebibliography}{99}
\bibitem{KonopelchenkoLandolfi}
Konopelchenko B G  and Landolfi G 1999 {Generalized Weierstrass representation for surfaces in 
multidimensional Riemann spaces} {\it J. Geom. Phys.} {\bf 29} 319-333.

\bibitem{KonopelchenkoTaimanov1996}
Konopelchenko B and Taimanov I 1996 {Constant mean curvature surfaces via an integrable dynamical system} 
{\it J. Phys. A: Math. Gen.} {\bf 29} 1261--5.

\bibitem{CK1996}
Carroll R and Konopelchenko B 1996 {Generalized Weierstrass-Enneper inducing conformal immersions and gravity} 
{\it Int. J. Mod. Phys. A} {\bf 11} 1183--216. 

\bibitem{Heleinbook}
H\'{e}lein F 2001 {\it Constant Mean Curvature Surfaces, Harmonic Maps and Integrable Systems} (Basel: Birkh\"{a}user).

\bibitem{GY1}
Grundland A M and Yurdu\c{s}en \.{I} 2009 {On analytic descriptions of two-dimensional surfaces associated with the 
$\mathbb{C}P^{N-1}$ models} {\it J. Phys. A: Math. Theor.} {\bf 42} 172001.

\bibitem{HYZ}
Hussin V, Yurdu\c{s}en \.{I} and Zakrzewski W J 2010 {Canonical surfaces associated with projectors in Grassmannian sigma
models} {\it J. Math. Phys.} {\bf 51} 103509.

\bibitem{PostG}
Post S and Grundland A M 2012 {Analysis of $\mathbb{C}P^{N-1}$ sigma models via projective structures} 
{\it Nonlinearity} {\bf 25} 1--36.

\bibitem{4solutionsG(24)}
Li Z Q and Yu Z H 1999 {Constant curved minimal 2-spheres in $G(2,4)$}
{\it Manuscripta Math.} {\bf 100} 305-16.

\bibitem{solutionsG(25)}
Jiao X X and Peng J G 2004 {Classification of holomorphic spheres of constant curvature in complex 
Grassmann manifold $G(2,5)$} 
{\it Diff. Geom. Appl.} {\bf 20} 267-77.

\bibitem{DHZ2013I}
Delisle L, Hussin V and Zakrzewski W J 2013 {Constant curvature solutions of {Grassmannian} sigma models: (1) 
{Holomorphic} solutions} {\it J. Geom. Phys.}  {\bf 66} 24-36.

\bibitem{DHZ2}
Delisle L, Hussin V and Zakrzewski W J 2013 {Constant curvature solutions of Grassmannian sigma models: 
(2) Non-holomorphic solutions} {\it J. Geom. Phys} {\bf 71} 1--10.


\bibitem{DHYZ}
Delisle L, Hussin V, Yurdu\c{s}en \.{I} and Zakrzewski W J 2015 {Constant curvature surfaces of the 
supersymmetric $\field{C}P^{N-1}$ sigma model} {\it J. Math. Phys.} {\bf 56} 023506-1--18.

\bibitem{DHZ5} Delisle L, Hussin V, Zakrzewski W J, 2016 {General construction of solutions of the supersymmetric $\field{C}P^{N-1}$ sigma model}, {\it J. Math. Phys}, {\bf 57}, 023506.

\bibitem{HLYZ}
Hussin V, Lafrance M, Yurdu\c{s}en \.{I} and Zakrzewski  W J 2018 {Holomorphic solutions of the susy 
Grassmannian $\sigma$-model and gauge invariance} {\it J. Phys. A: Math. Theor.} {\bf 51} 185401-1--8.

\bibitem{witten} 
Witten E 1977 {Supersymmetric form of the nonlinear $\sigma$ model in two dimensions} {\it Phys. Rev. D} 
{\bf 16} 2991-4.

\bibitem{susyG(MN)}
Din A M, Lukierski J and Zakrzewski W J 1982 {General classical-solutions of a supersymmetric non-linear coupled 
boson-fermion model in two dimensions} {\it Nucl. Phys. B} {\bf 194} 157-71.

\bibitem{fujii1}
Fujii K, Koikawa T and Sasaki R 1984 {Classical Solutions for the Supersymmetric
Grassmannian Sigma Models in Two Dimensions I} {\it Prog. Theor. Phys.} {\it 71} 388-94.

\bibitem{fujii2}
Fujii K and Sasaki R 1984 {Classical Solutions for the Supersymmetric
Grassmannian Sigma Models in Two Dimensions II} {\it Prog. Theor. Phys.} {\bf 71} 831-9.

\bibitem{Zbook}
Zakrzewski W J 1989 {\it Low Dimensional Sigma Models} (Bristol: Adam Hilger).

\bibitem{Hussin2006}
Hussin V and Zakrzewski W J 2006 {Susy $\mathbb{C}P^{N-1}$ model and surfaces in $R^{N^2-1}$} {\it J. Phys. A: Math. Gen.}
{\bf 39} 14231-14240.

\bibitem{MacFarlane}
MacFarlane A J 1978 {Generalizations of $\sigma$-models and $\field{C}P^N$ models and instantons} {\it Phys. Lett. B} 
{\bf 82} 239-41.


\bibitem{Calabi}
Calabi E 1953 {Isometric embedding of complex manifolds} {\it Ann. Math.} {\bf 58} 1-23.

\bibitem{Griffiths}
Griffiths P and Harris J 1978 {\it Principle of algebraic geometry} (Wiley, NY).

\bibitem{Yang}
Yang K 1989 {\it Complete and compact minimal surfaces} (Kluwer, Dordrecht).

\bibitem{solutionsG(2n)}
Peng C K and Xu X W 2015 {Classification of minimal homogeneous two-spheres in the complex Grassmann manifold
$G(2,n)$} {\it J. Math. Pures Appl.} {\bf 103} 374-99.













\end{thebibliography}
\end{document}